\documentclass[a4paper,english,copyright,creativecommons]{eptcs}
\usepackage[T1]{fontenc}
\usepackage[utf8]{inputenc}
\usepackage{babel}
\usepackage{verbatim}
\usepackage{calc}
\usepackage{textcomp}
\usepackage{amsthm}
\usepackage{amsmath}
\usepackage{amssymb}
\usepackage{stmaryrd}
\usepackage{wasysym}
\usepackage{xargs}[2008/03/08]
\ifx\hypersetup\undefined
  \AtBeginDocument{%
    \hypersetup{unicode=true}
  }
\else
  \hypersetup{unicode=true}
\fi

\makeatletter

\pdfpageheight\paperheight
\pdfpagewidth\paperwidth

\providecommand{\tabularnewline}{\\}

\theoremstyle{plain}
\newtheorem{thm}{\protect\theoremname}[section]
  \theoremstyle{definition}
  \newtheorem{defn}[thm]{\protect\definitionname}
  \theoremstyle{definition}
  \newtheorem{example}[thm]{\protect\examplename}
  \theoremstyle{plain}
  \newtheorem{lem}[thm]{\protect\lemmaname}

\providecommand{\event}{PxTP~2015} 
\usepackage{mathpartir}
\usepackage{adjustbox}

\title{Translating HOL to Dedukti}
\author{
Ali Assaf
\institute{Inria, Paris, France}
\institute{\'Ecole Polytechnique, Palaiseau, France}
\and
Guillaume Burel
\institute{ENSIIE/C\'edric, \'Evry, France}
}
\def\titlerunning{Translating HOL to Dedukti}
\def\authorrunning{A. Assaf and G. Burel}

\hypersetup{colorlinks=true}

\makeatother

\usepackage{listings}
  \providecommand{\definitionname}{Definition}
  \providecommand{\examplename}{Example}
  \providecommand{\lemmaname}{Lemma}
\providecommand{\theoremname}{Theorem}

\begin{document}
\global\long\def\nothing{\rule{0em}{0ex}}
\global\long\def\centerv#1{\adjustbox{valign=m}{\ensuremath{#1}}}

\global\long\def\lambdapi{\lambda\Pi}

\global\long\def\setname#1{\ensuremath{\mathcal{#1}}}
\global\long\def\metaname#1{\ensuremath{\mathit{#1}}}
\global\long\def\rulename#1{\ensuremath{\textsc{#1}}}
\global\long\def\constname#1{\ensuremath{\mathsf{#1}}}
\global\long\def\sortname#1{\mathsf{#1}}

\global\long\def\operator#1{\mathop{#1\nothing}}

\global\long\def\typesort{\sortname{Type}}
\global\long\def\kindsort{\sortname{Kind}}

\global\long\def\const#1{\constname{#1}}

\global\long\def\appterm#1#2{\mathop{#1\nothing}#2}

\global\long\def\lamterm#1#2#3{\lambda#1:#2.\,#3}

\global\long\def\prodterm#1#2#3{\Pi#1:#2.\,#3}

\global\long\def\ulam#1#2{\lambda#1.#2}

\global\long\def\arrow{\rightarrow}
\global\long\def\rewrite{\leadsto}
\global\long\def\steps{\longrightarrow}
\global\long\def\booleq{\Leftrightarrow}

\global\long\def\emptycontext{\cdot}
\global\long\def\rewriterule#1#2#3{\left[#1\right]\,#2\,\leadsto\,#3}

\global\long\def\transterm#1{\left|#1\right|}
\global\long\def\transtype#1{\left\Vert #1\right\Vert }

\global\long\def\wellformed{\ \metaname{well\mbox{-}formed}}
\global\long\def\context{\ \metaname{context}}
\global\long\def\signature{\ \metaname{signature}}
\global\long\def\freevars{\operator{\metaname{FV}}}

\global\long\def\D{\mathcal{D}}

\newcommandx\inferencerule[3][usedefault, addprefix=\global, 1=]{\centerv{\inferrule*[right=#1]{#2\nothing}{#3}}}

\global\long\def\ttypeconst{\const{type}}
\global\long\def\tboolconst{\const{bool}}
\global\long\def\tindconst{\const{ind}}
\global\long\def\tarrowconst{\const{arrow}}

\global\long\def\ttype{\ttypeconst}
\global\long\def\tbool{\tboolconst}
\global\long\def\tind{\indconst}
\global\long\def\tarrow#1#2{\appterm{\appterm{\tarrowconst}{#1}}{#2}}

\global\long\def\ttermconst{\const{term}}
\global\long\def\teqconst{\const{eq}}
\global\long\def\tselectconst{\const{select}}

\global\long\def\tterm#1{\appterm{\ttermconst}{#1}}
\global\long\def\teq#1#2#3{\appterm{\appterm{\appterm{\teqconst}{#1}}{#2}}{#3}}
\global\long\def\tselect#1{\appterm{\tselectconst}{#1}}

\global\long\def\tproofconst{\const{proof}}
\global\long\def\tproof#1{\appterm{\tproofconst}{#1}}

\global\long\def\topspace{\rule[0ex]{0em}{2.5ex}}

\global\long\def\botspace{\rule[-1ex]{0em}{0ex}}

\global\long\def\vbotspace#1{\rule[-#1ex]{0em}{0ex}}

\global\long\def\vtopspace#1{\rule{0em}{#1ex}}

\title{Translating HOL to Dedukti}
\maketitle
\begin{abstract}
Dedukti is a logical framework based on the $\lambdapi$-calculus
modulo rewriting, which extends the $\lambdapi$-calculus with rewrite
rules. In this paper, we show how to translate the proofs of a family
of HOL proof assistants to Dedukti. The translation preserves binding,
typing, and reduction. We implemented this translation in an automated
tool and used it to successfully translate the OpenTheory standard
library.
\end{abstract}

\section{\label{sec:introduction}Introduction}

Dedukti is a logical framework for defining logics and expressing
proofs in those logics \cite{boespflug_lambda-pi-calculus_2012}.
Following the LF legacy \cite{harper_framework_1993}, it is based
on the \emph{$\lambdapi$-calculus modulo rewriting}, which extends
the \emph{$\lambdapi$-calculus} with rewrite rules. Cousineau and
Dowek \cite{cousineau_embedding_2007} showed that functional \emph{pure
type systems} (PTS), a large class of calculi that are at the basis
of many proof systems, can be embedded in the $\lambdapi$-calculus
modulo rewriting in a way that is complete and that preserves reductions
(i.e.~program evaluation). This led to propose Dedukti as a universal
proof framework.

In this paper, we focus on translating the proofs of HOL to Dedukti.
HOL refers to a family of theorem provers built on a common logical
system known as \emph{higher-order logic} or \emph{simple type theory}
\cite{church_formulation_1940}. It includes systems such as HOL Light,
HOL4, and ProofPower-HOL. These systems are fairly popular and a large
number of important mathematical results have been formalized in them
\cite{hales_jordan_2007,hales_revision_2011,wiedijk_qed_2007}.

\subsubsection*{Universal proof checking}

Using Dedukti as a logical framework serves two goals. First, in the
short term, it serves as an alternative, independent proof checker,
providing an additional layer of confidence over each system. The
second, longer term goal, is interoperability. Proof systems are becoming
increasingly important, both in the formalization of mathematics and
in software engineering. However, they are usually developed separately,
with very little interoperability in mind. As a result, it is currently
very difficult to reuse a proof from one system in another one. Embedding
these different systems in a single unified framework is the first
step to bring them closer together, and opens the way for theory management
systems \cite{hurd_opentheory_2011,rabe_scalable_2011} to combine
their proofs in order to construct and verify larger theories.

\subsubsection*{The $\protect\lambdapi$-calculus as a logical framework}

The $\lambdapi$-calculus, also known as LF, is a typed $\lambda$-calculus
with dependent types. Through the \emph{Curry--Howard correspondence},
it can express a wide variety of logics \cite{harper_framework_1993}.
Several formalizations of HOL in LF have been proposed \cite{appel_foundational_2001,schurmann_executable_2006,rabe_representing_2010}.

\pagebreak{}

The main concept behind this correspondence is the ``\emph{propositions
as types}'' principle. Typically, we define a context declaring the
types, terms, and judgments of the original logic, in such a way that
provability in the logic corresponds to type inhabitation in the context.
For HOL, the signature would be:
\begin{flalign*}
 & \begin{array}{lll}
\ttypeconst & : & \typesort\\
\const{bool} & : & \ttype\\
\tarrowconst & : & \ttype\arrow\ttype\arrow\ttype
\end{array}\\
 & \begin{array}{lll}
\ttermconst & : & \ttype\arrow\typesort\\
\const{lam} & : & \left(\tterm{\alpha}\arrow\tterm{\beta}\right)\arrow\tterm{\left(\tarrow{\alpha}{\beta}\right)}\\
\const{app} & : & \tterm{\left(\tarrow{\alpha}{\beta}\right)}\arrow\tterm{\alpha}\arrow\tterm{\beta}
\end{array}\\
 & \begin{array}{lll}
\tproofconst & : & \tterm{\const{bool}}\\
\const{rule\_1} & : & \ldots\\
\const{rule\_2} & : & \ldots
\end{array}
\end{flalign*}
For each proposition $\phi$ of the logic, we assign a type $\transtype{\phi}$
in the $\lambdapi$-calculus. The provability of the proposition $\phi$
corresponds to the inhabitation of the type $\transtype{\phi}$. Similarly,
we translate proofs as terms inhabiting those types, and the correctness
of the proof corresponds to the well-typedness of the term.

However, because the $\lambdapi$-calculus does not have polymorphism,
we cannot translate propositions directly as types, as doing so would
prevent us from quantifying over propositions for example. Instead,
for each proposition $\phi$, we have two translations: one translation
$\transterm{\phi}$ \emph{as a term}, and another $\transtype{\phi}=\tproof{\transterm{\phi}}$
\emph{as a type}. This correspondence has been successfully used to
embed logics in the LF framework \cite{harper_framework_1993,geuvers_logical_1999},
implemented in Twelf \cite{pfenning_system_1999}.

\subsubsection*{The $\protect\lambdapi$-calculus vs. the $\protect\lambdapi$-calculus
modulo rewriting}

An important limitation of LF is that these encodings do not preserve
reduction (i.e.~program evaluation), and therefore it does not preserve
equivalence: if $M\equiv_{\beta}M'$ then $\transterm M\not\equiv_{\beta}\transterm{M'}$.
For example, the term $\appterm{(\lamterm x{\alpha}x)}x$ is encoded
as $\appterm{\appterm{\const{app}}{\left(\appterm{\const{lam}}{\left(\lamterm x{\tterm{\alpha}}x\right)}\right)}}x$
which is not equivalent to $x$. This is problematic not only because
it makes the representation larger and hence less efficient but also
because conversion proofs may be very long.

By extending the $\lambdapi$-calculus with rewrite rules such as
\[
\tterm{\left(\tarrow{\alpha}{\beta}\right)}\rewrite\tterm{\alpha}\arrow\tterm{\beta}\ ,
\]
we can identify the type $\tterm{\left(\tarrow{\alpha}{\beta}\right)}$
with the type $\tterm{\alpha}\arrow\tterm{\beta}$ and thus define
a translation that is lighter and that preserves reductions. The encoding
of the terms becomes more compact, as we represent $\lambda$-abstractions
by $\lambda$-abstractions, applications by applications, etc. For
example, the term $\appterm{(\lamterm x{\alpha}x)}x$ is encoded as
$\appterm{\left(\lamterm x{\tterm{\alpha}}x\right)}x$. Such an encoding
is impossible in LF for higher-order theories such as system~F, HOL,
or the calculus of constructions.

Moreover, our translation is modular enough so that we can extend
the notion of reduction to the proofs of HOL and recover the pure
type system nature of HOL \cite{barendregt_lambda_1992}. This might
be beneficial for several reasons:
\begin{enumerate}
\item It gives a reduction semantics for the proofs of HOL.
\item It allows compressing the proofs further by replacing conversion proofs
with reflexivity.
\item Several other proof systems (Coq, Agda, etc.) are based on pure type
systems, so expressing HOL as a PTS fits in the large scale of interoperability.
\end{enumerate}

\subsubsection*{HOL and OpenTheory}

The theorem provers of the HOL family (HOL Light, HOL4, ProofPower-HOL,
etc.) are built on a common logical formalism known as \emph{higher-order
logic}, and have fairly similar core implementations.

A recurrent issue when trying to retrieve proofs from these systems
is that they do not keep a trace of their proofs \cite{hurd_opentheory_2011,keller_importing_2010,obua_importing_2006}.
Following the LCF architecture, they represent their theorems using
an abstract datatype and thus guarantee their safety without the need
to remember their proofs. This approach reduces memory consumption
but hinders their ability to share proofs.

Fortunately, several proposals have already been made to solve this
problem \cite{hurd_opentheory_2011,obua_importing_2006}. Among them
is the OpenTheory project. It defines a standard format called the
\emph{article format} for recording and sharing HOL theorems. An article
file contains a sequence of elementary commands to reconstruct proofs.
Importing a theorem requires only a mechanical execution of the commands.

The format is limited to the HOL family, and cannot be used to communicate
the proofs of Coq for example. However, it is an excellent starting
point for our translation. Choosing OpenTheory as a front-end has
several advantages:
\begin{itemize}
\item We cover all the systems of the HOL family that can export their proofs
to OpenTheory with a single implementation. As of today, this includes
HOL Light, HOL4, and ProofPower-HOL.\footnote{Isabelle/HOL can currently read from but not write to OpenTheory.}
\item The implementation of a theorem prover can change, so the existence
of this standard, documented proof format is extremely helpful, if
not necessary.
\item The OpenTheory project also defines a large common standard theory
library, covering the development of common datatypes and mathematical
theories such as lists and natural numbers. This substantial body
of theories was used as a benchmark for our implementation.
\end{itemize}

\subsubsection*{Related work}

Several formalizations of HOL in LF have been proposed \cite{appel_foundational_2001,rabe_representing_2010,schurmann_executable_2006}.
To our knowledge, they lack an actual implementation of the translation.
Other translations have been proposed to automatically extract the
proofs of HOL to other systems such as Isabelle/HOL \cite{kaliszyk_scalable_2013,obua_importing_2006},
Nuprl \cite{naumov_hol/nuprl_2001}, or Coq \cite{keller_importing_2010}.
With the exception of the implementation of Kalyszyk and Krauss \cite{kaliszyk_scalable_2013},
these tools suffer from scalability problems. Our translation is lightweight
enough to be scalable and provides promising results. The implementation
of Kalyszyk and Krauss is the first efficient and scalable translation
of HOL Light proofs, but its target is Isabelle/HOL, a system that,
unlike Dedukti, is foundationally very close to HOL Light.

ProofCert \cite{chihani_foundational_2013} is another project like
Dedukti that aims at providing a universal framework for checking
proofs. Unlike Dedukti, it is based on sequent calculus. It can handle
linear, intuitionistic, and classical logics. To our knowledge, there
are no automated translations of systems like HOL to ProofCert that
have been implemented yet.

A project complementary to ours is Coqine \cite{boespflug_coqine:_2012},
which proposes a translation of the \emph{calculus of inductive constructions}
(CIC), the formalism behind Coq, to Dedukti. The translation has been
implemented in an automated tool that translates the proofs compiled
by Coq to Dedukti. It can handle most of the features of Coq, and
has been used to translate a part of its standard library.

\pagebreak{}

\subsubsection*{Contributions}

We define a translation of the types, terms and proofs of HOL to Dedukti.
We use the rewriting techniques of Cousineau and Dowek \cite{cousineau_embedding_2007}
to obtain a shallow embedding that is lightweight and modular. We
implemented this translation in an automated tool called Holide, which
automatically translates the proofs of HOL written in the OpenTheory
format to Dedukti. We used it to successfully translate the OpenTheory
standard library.

\subsubsection*{Outline}

The rest of this paper is organized as follows. Section \ref{sec:dedukti}
presents Dedukti and the $\lambdapi$-calculus modulo rewriting. Section
3 presents HOL and the logical system behind it. Section 4 defines
the translation of HOL to Dedukti. In Section 5, we show that the
translation is correct. Section 6 discusses the details of our implementation
and the results obtained by translating the OpenTheory standard library.
Section 7 discusses some additional applications of rewriting. Finally,
Section 8 summarizes and considers future work.

\section{\label{sec:dedukti}Dedukti}

Dedukti is essentially a type checker for the $\lambdapi$-calculus
modulo rewriting \cite{boespflug_lambda-pi-calculus_2012}, which
extends the $\lambdapi$-calculus with rewrite rules. We choose a
presentation based on pure type systems \cite{barendregt_lambda_1992},
which makes no syntactic distinction between terms, usually denoted
by $M$ or $N$, and types, usually denoted by $A$ or $B$.

We assume countably infinite sets of variables and constants. There
are two sorts, $\typesort$ and $\kindsort$. The sort $\typesort$
is the type of types and the sort $\kindsort$ is the type of $\typesort$.
We write $\lamterm xAM$ for abstractions and $\appterm MN$ for applications.
The type of functions is written $\prodterm xAB$, or $A\arrow B$
when $x$ does not appear free in $B$. Application is left-associative
while the arrow $\arrow$ is right-associative. Terms are considered
up to $\alpha$-equivalence. Contexts contain the types of variables
while signatures contain the types of constants and their rewrite
rules. Each rewrite rule is accompanied by a context $\Gamma$ to
ensure it is well-typed.
\begin{defn}
The syntax of the $\lambdapi$-calculus modulo rewriting is:
\[
\begin{array}{llcl}
\text{variables} & x,y\\
\text{constants} & c\\
\text{sorts} & s & ::= & \typesort\mid\kindsort\\
\text{terms} & M,N,A,B & ::= & x\mid c\mid s\mid\prodterm xAB\mid\lamterm xAM\mid M\,N\\
\text{contexts} & \Gamma,\Delta & ::= & \cdot\mid\Gamma,x:A\\
\text{signatures} & \Sigma & ::= & \cdot\mid\Sigma,c:A\mid\Sigma,\rewriterule{\Gamma}MN
\end{array}
\]

\end{defn}
If $R$ is a set of rewrite rules, we write $\steps_{R}$ for the
induced reduction relation, $\steps_{R}^{+}$ for its transitive closure,
$\steps_{R}^{*}$ for its reflexive transitive closure, and $\equiv_{R}$
for its reflexive symmetric transitive closure. Given a signature
$\Sigma$, we write $\beta\Sigma$ for the union of the $\beta$ rule
with the rewrite rules of $\Sigma$.

The typing judgments $\Sigma\mid\Gamma\vdash M:A$ are accompanied
by context formation judgments $\Sigma\mid\Gamma\context$ and signature
formation judgments $\Sigma\signature$. We write $\Gamma\vdash M:A$
and $\Gamma\context$ instead of $\Sigma\mid\Gamma\vdash M:A$ and
$\Sigma\mid\Gamma\context$ when the signature is not ambiguous. The
rules are presented in Figure~\ref{fig:typing-lambdapi}.
\begin{figure}
\fbox{\begin{minipage}[t]{1\columnwidth}%
\begin{mathpar}
\inferrule*[right=Var]{\Gamma \context \\ (x : A) \in \Gamma}{\Gamma \vdash x : A}

\inferrule*[right=Const]{\Gamma \context \\ (c : A) \in \Sigma}{\Gamma \vdash c : A}

\inferrule*[right=Type]{\Gamma \context}{\Gamma \vdash \typesort : \kindsort}

\inferrule*[right=Prod]{\Gamma \vdash A : \typesort \\ \Gamma, x : A \vdash B : s}{\Gamma \vdash \prodterm{x}{A}{B} : s}

\inferrule*[right=Abs]{\Gamma \vdash A : \typesort \\ \Gamma, x : A \vdash M : B}{\Gamma \vdash \lamterm{x}{A}{M} : \prodterm{x}{A}{B}}

\inferrule*[right=App]{\Gamma \vdash M : \prodterm{x}{A}{B} \\ \Gamma \vdash N : A}{\Gamma \vdash \appterm{M}{N} : [N/x]B}

\inferrule*[right=Conv]{\Gamma \vdash M : A \\ \Gamma \vdash B : \typesort \\ A \equiv_{\beta\Sigma} B}{\Gamma \vdash M : B}\\

\inferrule*[right=EmptyCtx]{\Sigma \signature}{\emptycontext \context}

\inferrule*[right=VarCtx]{\Gamma \vdash A:\typesort \\ x \not\in \Gamma}{\Gamma,x:A \context}\\

\inferrule*[right=EmptySig]{ }{\emptycontext \signature}

\inferrule*[right=ConstSig]{\Sigma\mid\emptycontext \vdash A:s \\ c \not\in \Sigma}{\Sigma,c:A \signature}

\inferrule*[right=RewriteSig]{\Sigma\mid\Gamma \vdash M:A \\ \Sigma\mid\Gamma \vdash N:A}{\Sigma,\rewriterule{\Gamma}{M}{N} \signature}\\
\end{mathpar}%
\end{minipage}}

\protect\caption{\label{fig:typing-lambdapi}Typing rules of the $\protect\lambdapi$-calculus}
\end{figure}

\global\long\def\natconst{\const{nat}}
\global\long\def\zeroconst{\const z}
\global\long\def\succconst{\const s}
\global\long\def\streamconst{\const{stream}}

\global\long\def\natterm{\natconst}
\global\long\def\zeroterm{\zeroconst}
\global\long\def\succterm#1{\appterm{\succconst}{#1}}
\global\long\def\streamterm{\streamconst}

\begin{example}
Let $\Sigma$ be the signature containing
\[
\alpha:\typesort,c:\alpha,f:\alpha\to\typesort
\]
 and the rewrite rule
\[
\rewriterule{\emptycontext}{\appterm fc}{\prodterm y{\alpha}{\appterm fy\to\appterm fy}}\,.
\]
The term $\lamterm x{\appterm fc}{\appterm{\appterm xc}x}$ is well-typed
in $\Sigma$ and has the type $\appterm fc\arrow\appterm fc$. Notice
that this term would not be well-typed without the rewrite rule, even
if we replace all occurences of $\appterm fc$ by $\prodterm y{\alpha}{\appterm fy\to\appterm fy}$.
\end{example}
Dedukti imposes some additional restrictions on the rewrite rules
to keep type-checking decidable. In particular, the left side of a
rewrite rule must belong to the higher-order pattern fragment \cite{miller_unification_1991,miller_proofs_2004}
and the free variables of the right side must appear on the left side.
Moreover, the reduction relation $\steps_{\beta\Sigma}$ should be
confluent and strongly normalizing. This property is not verified
by the system and it is up to the user to ensure that it is indeed
the case. We discuss this in Section \ref{sub:correctness}.

\section{\label{sec:hol}HOL}

There are many different formulations for higher-order logic. The
intuitionistic formulation is based on implication and universal quantification
as primitive connectives, but the current systems generally use a
formulation called $\mathrm{Q}_{0}$ \cite{andrews_introduction_1986}
based on equality as a primitive connective. We take as reference
the logical system used by OpenTheory \cite{hurd_opentheory_2011},
which we will now briefly present.

The terms of the logic are terms of the simply typed $\lambda$-calculus,
with a base type $\const{bool}$ representing the type of propositions
and a type $\const{ind}$ of individuals. The terms can contain constant
symbols such as $(\const =)$, the symbol for equality, or $\const{select}$,
the symbol of choice. The logic supports a restricted form of polymorphism,
known as ML-style polymorphism, by allowing type variables, such as
$\alpha$ or $\beta$, to appear in types. For example, the type of
$(=)$ is $\alpha\arrow\alpha\arrow\const{bool}$.

Types can be parameterized through type operators of the form $p(A_{1},\ldots,A_{n})$.
For example, $\const{list}$ is a type operator of arity 1, and $\const{list}(\const{bool})$
is the type of lists of booleans. Type variables and type operators
are enough to describe all the types of HOL, because $\const{bool}$
can be seen as a type operator of arity 0, and the arrow $\arrow$
as a type operator of arity 2. Hence the type of $(=_{\alpha})$ is
in fact $\arrow(\alpha,\arrow(\alpha,\const{bool}()))$. We still
write $A\arrow B$ instead of $\const{\arrow}(A,B)$ for arrow types,
$p$ instead of $p()$ for type operators of arity 0, and $M=N$ instead
of $\appterm{\appterm{(=)}M}N$ when it is more convenient.
\begin{defn}
The syntax of HOL is:
\[
\begin{array}{llcl}
\text{type variables} & \alpha,\beta\\
\text{type operators} & p\\
\text{types} & A,B & ::= & \alpha\mid p(A_{1},\ldots,A_{n})\\
\text{term variables} & x,y\\
\text{term constants} & c\\
\mbox{terms} & M,N & ::= & x\mid\lamterm xAM\mid\appterm MN\mid c
\end{array}
\]

\end{defn}
The propositions of the logic are the terms of type $\const{bool}$
and the predicates are the terms of type $A\arrow\const{bool}$. We
use letters such as $\phi$ or $\psi$ to denote propositions. The
contexts, denoted by $\Gamma$ or $\Delta$, are sets of propositions,
and the judgments of the logic are of the form $\Gamma\vdash\phi$.
The derivation rules are presented in Figure \ref{fig:derivation-rules-hol}.

\begin{figure}
\fbox{\begin{minipage}[t]{1\columnwidth}%
\begin{mathpar}
\inferrule*[right=Refl $M$]{ }{\vdash M = M}

\inferrule*[right=AbsThm $x$]{\Gamma \vdash M = N}{\Gamma \vdash \lamterm{x}{A}{M}=\lamterm{x}{A}{N}}

\inferrule*[right=AppThm]{\Gamma \vdash F = G \\ \Delta \vdash M = N}{\Gamma \cup \Delta \vdash \appterm{F}{M} = \appterm{G}{N}}

\inferrule*[right=Beta $x$ $M$]{ }{\vdash \appterm{(\lamterm{x}{A}{M})}{x} = M}

\inferrule*[right=Assume]{ }{\{\phi\} \vdash \phi}

\inferrule*[right=EqMp]{\Gamma \vdash \phi = \psi \\ \Delta \vdash \phi}{\Gamma \cup \Delta \vdash \psi}

\inferrule*[right=DeductAntiSym]{\Gamma \vdash \phi \\ \Delta \vdash \psi}{(\Gamma - \{\psi\}) \cup (\Delta - \{\phi\}) \vdash  \phi = \psi}

\inferrule*[right=Subst $\sigma$]{\Gamma \vdash \phi}{\Gamma[\sigma] \vdash \phi[\sigma]}

\\
\end{mathpar}%
\end{minipage}}

\protect\caption{\label{fig:derivation-rules-hol}Derivation rules of HOL}
\end{figure}

\begin{example}
\label{ex:hol-transitivity}Here is a derivation of the transitivity
of equality: if $\Gamma\vdash x=y$ and $\Delta\vdash y=z$, then
$\Gamma\cup\Delta\vdash x=z$.\begin{mathpar}
\inferrule*[right=EqMp]{
	\inferrule*[right=AppThm]{
		\inferrule*[right=Refl]{ }{\vdash (\appterm{(=)}{x}) = (\appterm{(=)}{x})} \\
		\Delta \vdash y = z}
		{\Delta \vdash (x = y) = (x = z)} \\
	\Gamma \vdash x = y}
	{\Gamma \cup \Delta \vdash x = z}
\end{mathpar}
\end{example}
HOL supports mechanisms for defining new types and constants in a
conservative way. We will not consider them here. In addition to the
core derivation rules, three axioms are assumed:
\begin{itemize}
\item $\eta$-equality, which states that $\lamterm xA{\appterm Mx}=M$,
\item the axiom of choice, with a predeclared symbol of choice called $\const{select}$,
\item the axiom of infinity, which states that the type $\const{ind}$ is
infinite.
\end{itemize}
It is important to note that from $\eta$-convertibility and the axiom
of choice, we can derive the excluded middle \cite{beeson_foundations_1985},
making HOL a classical logic.

\section{\label{sec:translation}Translation}

In this section we show how to translate HOL to Dedukti. We define
a signature $\Sigma$ containing primitive declarations and definitions,
and a translation function assigning, to every construct of the logic,
a term that is well-typed in the signature $\Sigma$.

\subsubsection*{HOL Types}

To translate the simple types of HOL, we declare a new Dedukti type
called $\const{type}$ and three constructors $\const{bool}$, $\const{ind}$
and $\const{arrow}$.
\[
\begin{array}{lll}
\ttypeconst & : & \typesort\\
\tboolconst & : & \ttype\\
\tindconst & : & \ttype\\
\tarrowconst & : & \ttype\arrow\ttype\arrow\ttype
\end{array}
\]
One should not confuse $\const{type}$, which is the type of Dedukti
terms that represent HOL types, with $\const{Type}$, which is the
type of Dedukti types. The translation of a HOL type as a Dedukti
term is defined inductively on the structure of the type.
\begin{defn}[Translation of a HOL type as a Dedukti term]
 For any HOL type $A$, we define $\transterm A$, the translation
of $A$ as a term, to be
\[
\begin{array}{lcl}
\transterm{\alpha} & = & \alpha\\
\transterm{\const{bool}} & = & \const{bool}\\
\transterm{\const{ind}} & = & \const{ind}\\
\transterm{A\arrow B} & = & \appterm{\const{arrow}\transterm A}{\transterm B}\ .
\end{array}
\]
More generally, if we have an $n$-ary HOL type operator $p$, we
declare a constant $p$ of type $\underbrace{\const{type}\arrow\ldots\arrow\const{type}}_{n}\arrow\const{type}$,
and we translate an instance $p\left(A_{1},\ldots,A_{n}\right)$ of
this type operator to the term $\appterm{\appterm{\appterm p{\transterm{A_{1}}}}{\cdots}}{\transterm{A_{n}}}$.
\end{defn}

\subsubsection*{HOL Terms}

We declare a new dependent type called $\const{term}$ indexed by
a $\const{type}$, and we identify the terms of type $\appterm{\const{term}}{(\appterm{\appterm{\const{arrow}}A}B)}$
with the functions of type $\appterm{\const{term}}A\arrow\appterm{\const{term}}B$
by adding a rewrite rule. We also declare a constant $\const{eq}$
for HOL equality and a constant $\const{select}$ for the choice operator.
\[
\begin{array}{lll}
\ttermconst & : & \ttype\arrow\typesort\\
\teqconst & : & \prodterm{\alpha}{\ttype}{\tterm{\left(\tarrow{\alpha}{\left(\tarrow{\alpha}{\tbool}\right)}\right)}}\\
\tselectconst & : & \prodterm{\alpha}{\ttype}{\tterm{\left(\tarrow{\left(\tarrow{\alpha}{\tbool}\right)}{\alpha}\right)}}
\end{array}
\]
\[
\rewriterule{\alpha:\ttype,\beta:\ttype}{\tterm{\left(\tarrow{\alpha}{\beta}\right)}}{\tterm{\alpha}\arrow\tterm{\beta}}
\]
The symbol $\const{term}$ can be seen as a decoding function that
assigns a Dedukti type to every HOL type. The translation of a term
$M$ of type $A$ will then be a term of type $\const{term}\transterm A$.
\begin{defn}[Translation of a HOL type as a Dedukti type]
 For any HOL type $A$, we define
\[
\transtype A=\tterm{\transterm A}.
\]

\end{defn}

\begin{defn}[Translation of a HOL term as a Dedukti term]
 For any HOL term $M$, we define $\transterm M$, the translation
of $M$ as a term to be
\[
\begin{array}{lcl}
\transterm x & = & x\\
\transterm{\appterm MN} & = & \appterm{\transterm M}{\transterm N}\\
\transterm{\lamterm xAM} & = & \lamterm x{\transtype A}{\transterm M}\\
\transterm{(=_{A})} & = & \const{eq}\transterm A\\
\transterm{\const{select}_{A}} & = & \appterm{\const{select}}{\transterm A}\ .
\end{array}
\]
More generally, for every HOL constant $c$ of type $A$, if $\alpha_{1},\ldots,\alpha_{n}$
are the free type variables that appear in $A$, we declare a new
constant $c$ of type
\[
\prodterm{\alpha_{1}}{\const{type}}{\ldots\prodterm{\alpha_{n}}{\const{type}}{\transtype A}}
\]
and we translate an instance $c_{A_{1},\ldots,A_{n}}$ of this constant
by the term $\appterm{\appterm{\appterm c{\transterm{A_{1}}}}{\cdots}}{\transterm{A_{n}}}$.\end{defn}
\begin{example}
The term $\appterm{(\lamterm x{\alpha}x)}x$ is translated to
\[
\begin{array}{ccc}
\transterm{\appterm{(\lamterm x{\alpha}x)}x} & = & \appterm{(\lamterm x{\appterm{\const{term}}{\alpha}}x)}x\end{array}
\]
which is convertible to $x$.
\end{example}

\subsubsection*{HOL Proofs}

We declare a new type $\tproofconst$, to express the proof judgments
of HOL. It is a dependent type, indexed by the proposition $\phi$
that it is proving.
\[
\begin{array}{lll}
\tproofconst & : & \tterm{\tbool}\to\typesort\end{array}
\]

\begin{defn}[Translation of HOL propositions as Dedukti types]
 For any HOL proposition $\phi$ (i.e.~a HOL term of type bool),
we define
\[
\transtype{\phi}=\tproof{\transterm{\phi}}.
\]
For any HOL context $\Gamma=\phi_{1},\ldots,\phi_{n}$, we define
\[
\transtype{\Gamma}=h_{\phi_{1}}:\transtype{\phi_{1}},\ldots,h_{\phi_{n}}:\transtype{\phi_{n}}
\]
 where $h_{\phi_{1}},\ldots,h_{\phi_{n}}$ are fresh variables.
\end{defn}
We now take care of the derivation rules of HOL (Figure \ref{fig:derivation-rules-hol}).
In the following, we write $\prodterm{x,y}AB$ as a shortcut for $\prodterm xA{\prodterm yAB}$.

\subsubsection*{Equality proofs}

We declare $\const{Refl}$, $\const{FunExt}$, and $\const{AppThm}$:
\[
\begin{array}{lll}
\const{Refl} & : & \prodterm{\alpha}{\ttype}{\prodterm x{\tterm{\alpha}}{\tproof{\left(\teq{\alpha}xx\right)}}}\\
\const{FunExt} & : & \prodterm{\alpha,\beta}{\ttype}{\prodterm{f,g}{\tterm{\left(\tarrow{\alpha}{\beta}\right)}}{}}\\
 &  & \left(\prodterm x{\tterm{\alpha}}{\tproof{\left(\teq{\beta}{\left(\appterm fx\right)}{\left(\appterm gx\right)}\right)}}\right)\arrow\tproof{\left(\teq{\left(\tarrow{\alpha}{\beta}\right)}fg\right)}\\
\const{AppThm} & : & \prodterm{\alpha,\beta}{\ttype}{\prodterm{f,g}{\tterm{\left(\tarrow{\alpha}{\beta}\right)}}{\prodterm{x,y}{\tterm{\alpha}}{}}}\\
 &  & \tproof{\left(\teq{\left(\tarrow{\alpha}{\beta}\right)}fg\right)}\arrow\tproof{\left(\teq{\alpha}xy\right)}\arrow\tproof{\left(\teq{\beta}{\left(\appterm fx\right)}{\left(\appterm gy\right)}\right)}
\end{array}
\]
The constant $\const{FunExt}$ corresponds to \emph{functional extensionality},
which states that if two functions $f$ and $g$ of type $A\arrow B$
are equal on all values $x$ of type $A$, then $f$ and $g$ are
equal. We can use it to translate both the $\rulename{AbsThm}$ rule
and the $\eta$ axiom. Finally, since our encoding is shallow, $\beta$-equality
can be proved by reflexivity.
\begin{defn}
The rules $\rulename{Refl}$, $\rulename{AbsThm}$, $\rulename{AppThm}$,
and $\rulename{Beta}$ are translated to
\[
\transterm{\inferencerule[Refl]{}{\vdash M=M}}=\appterm{\appterm{\const{Refl}}{\transterm A}}{\transterm M}\quad\text{(where \ensuremath{A\ }is the type of \ensuremath{M})}
\]
\[
\transterm{\inferencerule[AbsThm]{\D}{\Gamma\vdash\lamterm xAM=\lamterm xAN}}=\appterm{\appterm{\appterm{\const{FunExt}}{\transterm A}}{\transterm B}}{\transterm{\lamterm xAM}\,\transterm{\lamterm xAN}\,(\lamterm x{\transterm A}{\transterm{\D}})}
\]
\[
\transterm{\inferencerule[AppThm]{\D_{1}\qquad\D_{2}}{\Gamma\cup\Delta\vdash\appterm FM=\appterm GN}}=\appterm{\appterm{\appterm{\appterm{\appterm{\appterm{\appterm{\appterm{\const{AppThm}}{\transterm A}}{\transterm B}}{\transterm F}}{\transterm G}}{\transterm M}}{\transterm N}}{\transterm{\D_{1}}}}{\transterm{\D_{2}}}
\]
\[
\transterm{\inferencerule[Beta]{}{\appterm{(\lamterm xAM)}x=M}}=\appterm{\appterm{\const{Refl}}{\transterm B}}{\transterm M}\quad\mbox{(where \ensuremath{B\ }is the type of \ensuremath{M})}\ .
\]

\end{defn}

\subsubsection*{Boolean proofs}

We declare the constants $\const{PropExt}$ and $\const{EqMp}$:
\[
\begin{array}{lll}
\const{PropExt} & : & \prodterm{p,q}{\tterm{\tbool}}{}\\
 &  & \left(\tproof q\arrow\tproof p\right)\arrow\left(\tproof q\arrow\tproof p\right)\arrow\tproof{\left(\teq{\tbool}pq\right)}\\
\const{EqMp} & : & \prodterm{p,q}{\tterm{\tbool}}{\tproof{\left(\teq{\tbool}pq\right)}\arrow\tproof p\arrow\tproof q}
\end{array}
\]
The constant $\const{PropExt}$ corresponds to \emph{propositional
extensionality} and, together with $\const{EqMp}$, states that equality
on booleans in HOL behaves like the connective ``\emph{if and only
if}''.
\begin{defn}
The rules $\rulename{Assume}$, $\rulename{DeductAntiSym}$, and $\rulename{EqMp}$
are translated to
\[
\transterm{\inferencerule[Assume]{}{\{\phi\}\vdash\phi}}=h_{\phi}\quad\text{(where \ensuremath{h_{\phi}\ }is a fresh variable)}
\]
\begin{multline*}
\transterm{\inferencerule[DeductAntiSym]{\D_{1}\qquad\D_{2}}{(\Gamma-\{\psi\})\cup(\Delta-\{\phi\})\vdash\phi=\psi}}=\\
\appterm{\appterm{\appterm{\appterm{\const{PropExt}}{\transterm{\phi}}}{\transterm{\psi}}}{\left(\lamterm{h_{\psi}}{\transtype{\psi}}{\transterm{\D_{1}}}\right)}}{\left(\lamterm{h_{\phi}}{\transtype{\phi}}{\transterm{\D_{2}}}\right)}
\end{multline*}
\[
\transterm{\inferencerule[EqMp]{\D_{1}\qquad\D_{2}}{\Gamma\cup\Delta\vdash\psi}}=\appterm{\appterm{\appterm{\appterm{\const{EqMp}}{\transterm{\phi}}}{\transterm{\psi}}}{\transterm{\D_{1}}}}{\transterm{\D_{2}}}\ .
\]

\end{defn}

\subsubsection*{Substitution proofs}

The HOL rule $\rulename{subst}$ derives $\Gamma[\sigma]\vdash\phi[\sigma]$
from $\Gamma\vdash\phi$. In OpenTheory, the substitution can substitute
for both term and type variables but type variables are instantiated
first. For the sake of clarity, we split this rule in two steps: one
for term substitution of the form $\sigma=M_{1}/x_{1},\ldots,M_{n}/x_{n}$,
and one for type substitution of the form $\theta=A_{1}/\alpha_{1},\ldots,A_{m}/\alpha_{m}$.
In Dedukti, we have to rely on $\beta$-reduction to express substitution.
We can correctly translate a parallel substitution $M[M_{1}/x_{1},\ldots,M_{n}/x_{n}]$
as
\[
\appterm{\appterm{\appterm{(\lamterm{x_{1}}{B_{1}}{\ldots\lamterm{x_{n}}{B_{n}}M})}{M_{1}}}{\ldots}}{M_{n}}
\]
where $B_{i}$ is the type of $M_{i}$.
\begin{defn}
The rule $\rulename{Subst}$ is translated to
\[
\transterm{\inferencerule[TypeSubst]{\D}{\Gamma[\theta]\vdash\phi[\theta]}}=\appterm{\appterm{\appterm{(\lamterm{\alpha_{1}}{\const{type}}{\ldots\lamterm{\alpha_{m}}{\const{type}}{\transterm{\D}}})}{\transterm{A_{1}}}}{\ldots}}{\transterm{A_{m}}}
\]
\[
\transterm{\inferencerule[TermSubst]{\D}{\Gamma[\sigma]\vdash\phi[\sigma]}}=\appterm{\appterm{\appterm{(\lamterm{x_{1}}{\transtype{B_{1}}}{\ldots\lamterm{x_{n}}{\transtype{B_{n}}}{\transterm{\D}}})}{\transterm{M_{1}}}}{\ldots}}{\transterm{M_{n}}}
\]

\end{defn}

\section{\label{sub:correctness}Correctness}

\looseness=-1 The correctness of the translation is expressed by
two properties: \emph{completeness} and \emph{soundness}. The first
states that all the generated terms have the correct type. For example,
the translation of a term of type $A$ has type $\transtype A$ while
the translation of a proof of $\phi$ has type $\transtype{\phi}$.
The second states that if a proof term is well-typed in Dedukti, then
the proof is correct in the original logic. These two properties ensure
that we can use Dedukti as an independent proof checker: we can use
it to re-verify the proofs of OpenTheory, and moreover we can be sure
that, if a proof is accepted by Dedukti, then it is also valid in
OpenTheory.

\subsubsection*{Completeness}

Let $\Sigma$ be the signature of HOL containing the declarations
and rewrite rules of the previous sections.
\begin{lem}
\label{lem:hol-type-term}For any HOL type $A$,
\[
\Sigma\mid\alpha_{1}:\const{type},\ldots,\alpha_{n}:\const{type}\vdash\transterm A:\const{type}
\]
where $\alpha_{1},\ldots,\alpha_{n}$ are the free type variables
appearing in $A$.
\end{lem}

\begin{lem}
\label{lem:hol-term-term}For any HOL term $M$ of type $A$,
\[
\Sigma\mid\alpha_{1}:\const{type},\ldots,\alpha_{n}:\const{type},x_{1}:\transtype{A_{1}},\ldots x_{n}:\transtype{A_{n}}\vdash\transterm M:\transtype A
\]
where $\alpha_{1},\ldots,\alpha_{n}$ are the free type variables
and $x_{1}:A_{1},\ldots,x_{n}:A_{n}$ are the free term variables
appearing in $M$.
\end{lem}

\begin{thm}
\label{lem:hol-proof-term}For any HOL proof $\mathcal{D}$ of $\Gamma\vdash\phi$,
\[
\Sigma\mid\alpha_{1}:\const{type},\ldots,\alpha_{n}:\const{type},x_{1}:\transtype{A_{1}},\ldots x_{n}:\transtype{A_{n}},\transtype{\Gamma}\vdash\transterm{\mathcal{D}}:\transtype{\phi}
\]
where $\alpha_{1},\ldots,\alpha_{n}$ are the free type variables
and $x_{1}:A_{1},\ldots,x_{n}:A_{n}$ are the free term variables
appearing in $\D$.\end{thm}
\begin{proof}
By induction on the structure of $\mathcal{D}$.
\end{proof}

\subsubsection*{Soundness}

Proving the soundness of the embedding is less straightforward than
proving completeness. In fact, it is closely related to the confluence
and normalization properties of the system. We state the results here
and refer the reader to the works of Assaf, Cousineau, and Dowek \cite{assaf_conservativity_2013,cousineau_embedding_2007,dowek_models_2014}
for the complete proofs.\footnote{The terms \emph{soundness }and \emph{completeness }are interchanged
in Cousineau and Dowek's paper \cite{cousineau_embedding_2007}.}
\begin{lem}
The reduction relation $\steps_{\beta\Sigma}$ is confluent.
\end{lem}

\begin{lem}
The reduction relation $\steps_{\beta\Sigma}$ is strongly normalizing.
\end{lem}

\begin{thm}
If $\Sigma\mid\transtype{\Gamma}\vdash M:\transtype A$ then $M$
corresponds to a valid proof of $\Gamma\vdash A$ in HOL.
\end{thm}

\section{\label{sec:implementation}Implementation}

\looseness=-1 We implemented our translation in an automated tool
called Holide. It works as an OpenTheory virtual machine that additionally
keeps track of the corresponding proof terms for theorems. The program
reads a HOL proof written in the OpenTheory article format (\texttt{.art})
and outputs a Dedukti file (\texttt{.dk}) containing its translation.
We can run Dedukti on the generated file to verify it. All generated
files are linked with a hand-written file \texttt{hol.dk} containing
the signature $\Sigma$ that we defined in Section \ref{sec:translation}.
Our software is available online at \href{https://www.rocq.inria.fr/deducteam/Holide/}{https://www.rocq.inria.fr/deducteam/Holide/}.

\looseness=-1 HOL proofs are known to be very large \cite{kaliszyk_scalable_2013,keller_importing_2010,obua_importing_2006},
and we needed to implement sharing of proofs, terms, and types in
order to reduce them to a manageable size. OpenTheory already provides
some form of proof sharing but we found it easier to completely factorize
the derivations into individual steps.

We used Holide to translate the OpenTheory standard library. The library
is organized into logical packages, each corresponding to a theory
such as lists or natural numbers. We were able to verify all of the
generated files. The results are summarized in Table \ref{tab:translation-stdlib}.
We list the size of both the source files and the files generated
by the translation after compression using gzip. The reason we use
the size of the compressed files for comparison is because it provides
a more reasonable measure that is less affected by syntax formatting
and whitespace. We also list the time it takes to translate and verify
each package. These tests were done on a 64-bit Intel Xeon(R) CPU
@ 2.67GHz \texttimes{} 4 machine with 4 GB of RAM.

\begin{table}[t]
\begin{centering}
\begin{tabular}{lrrrr}
\hline
\noalign{\vskip\doublerulesep}
Package & \multicolumn{2}{c}{Size (kB)} & \multicolumn{2}{c}{Time (s)}\tabularnewline[\doublerulesep]
\cline{2-5}
\noalign{\vskip\doublerulesep}
 & OpenTheory & Dedukti & Translation & Verification\tabularnewline[\doublerulesep]
\hline
\noalign{\vskip\doublerulesep}
unit & 5 & 13 & 0.2 & 0.0\tabularnewline
function & 16 & 53 & 0.3 & 0.2\tabularnewline
pair & 38 & 121 & 0.8 & 0.5\tabularnewline
bool & 49 & 154 & 0.9 & 0.5\tabularnewline
sum & 84 & 296 & 2.1 & 1.1\tabularnewline
option & 93 & 320 & 2.2 & 1.2\tabularnewline
relation & 161 & 620 & 4.6 & 2.8\tabularnewline
list & 239 & 827 & 5.7 & 3.2\tabularnewline
real & 286 & 945 & 6.5 & 3.1\tabularnewline
natural & 343 & 1065 & 6.8 & 3.2\tabularnewline
set  & 389 & 1462 & 10.2 & 5.8\tabularnewline[\doublerulesep]
\hline
\noalign{\vskip\doublerulesep}
\textbf{Total} & 1702 & 5877 & 40.3 & 21.6\tabularnewline[\doublerulesep]
\hline
\end{tabular}
\par\end{centering}

\vspace{-2pt}

\protect\caption{\label{tab:translation-stdlib}Translation of the OpenTheory standard
library}
\end{table}

Overall, the size of the generated files is about 3 to 4 times larger
than the source files. Given that this is an encoding in a logical
framework, an increase in the size is to be expected, and we find
that this factor is very reasonable. There are no similar translations
to compare to except the one of Keller and Werner \cite{keller_importing_2010}.
The comparison is difficult because they work with a slightly different
form of input, but they produce several hundred megabytes of proofs.
Similary, an increase in verification time is to be expected compared
to verifying OpenTheory directly, but our results are still very reasonable
given the nature of the translation. Our time is about 4 times larger
than OpenTheory, which takes about 5 seconds to verify the standard
library. It is in line with the scalable translation of Kalyszyk and
Krauss to Isabelle/HOL, which takes around 30 seconds \cite{kaliszyk_scalable_2013}.
In comparison, Keller and Werner's translation takes several hours,
although we should note that our work greatly benefited from their
experience.

\section{\label{sec:extensions}Extensions}

In this section we show some additional advantages of having a translation
which preserves reduction.

\subsubsection*{Compressing conversion proofs}

One of the reasons why HOL proofs are so large is that conversion
proofs have to traverse the terms using the congruence rules $\rulename{AbsThm}$
and $\rulename{AppThm}$. Since we now prove $\beta$-reduction using
reflexivity, large conversion proofs could be reduced to a single
reflexivity step, therefore reducing the size of the proofs.\footnote{This also applies to conversions involving constant definitions, which
we did not cover here but are also assumed as an axiom in HOL.}
\begin{example}
The following proof of $\appterm f{(\appterm g{(\appterm{(\lamterm xAx)}x)})}=\appterm f{(\appterm g{(x}))}$,
\begin{mathpar}
\inferrule*{
	\inferrule*[right=Refl $f$]{ }{\vdash f=f} \\
	\inferrule*[right=AppThm]{
		\inferrule*[right=Refl $g$]{ }{\vdash g=g} \\
		\inferrule*[right=Beta]{ }{\vdash \appterm{(\lamterm{x}{A}{x})}{x} = x}
	} {\vdash \appterm{g}{(\appterm{(\lamterm{x}{A}{x})}{x})} = \appterm{g}{x}}
} {\vdash \appterm{f}{(\appterm{g}{(\appterm{(\lamterm{x}{A}{x})}{x})})} = \appterm{f}{(\appterm{g}{x})}}
\end{mathpar}can be translated simply as $\appterm{\appterm{\const{Refl}}C}{(\appterm f{(\appterm gx)})}$,
where $A\to B$ is the type of $g$ and $B\to C$ is the type of $f$.
\end{example}

\subsubsection*{HOL as a pure type system}

It turns out that HOL can be seen as a pure type system called $\lambda_{HOL}$
with three sorts \cite{barendregt_lambda_1992,geuvers_logics_1993}.
This formulation corresponds to intuitionistic higher-order logic.
However, this structure is lost in the $\mathrm{Q}_{0}$ formulation
used by the HOL systems. Our shallow embedding can be adapted to recover
this structure, and thus obtain a constructive and computational version
of HOL.

Instead of equality, we declare implication and universal quantification
as primitive connectives, and we define what provability means through
rewriting.
\[
\begin{array}{lll}
\const{imp} & : & \tterm{\left(\tarrow{\tbool}{\left(\tarrow{\tbool}{\tbool}\right)}\right)}\\
\const{forall} & : & \prodterm{\alpha}{\ttype}{\tterm{\left(\tarrow{\left(\tarrow{\alpha}{\tbool}\right)}{\tbool}\right)}}
\end{array}
\]
\[
\begin{array}{llll}
\left[p:\tterm{\tbool},q:\tterm{\tbool}\right] & \tproof{\left(\appterm{\appterm{\const{imp}}p}q\right)} & \rewrite & \tproof p\arrow\tproof q\\
\left[\alpha:\ttype,p:\tterm{\left(\tarrow{\alpha}{\tbool}\right)}\right] & \tproof{\left(\appterm{\const{forall}}p\right)} & \rewrite & \prodterm x{\tterm{\alpha}}{\tproof{\left(\appterm px\right)}}
\end{array}
\]
However, this time we do not even need to declare constants like $\const{Refl}$
and $\const{AppThm}$ for the derivation rules, because they are derivable.
Here is a derivation of the introduction and elimination rules for
implication for example:
\[
\begin{array}{lcl}
\const{imp\_intro} & : & \prodterm{p,q}{\tterm{\tbool}}{\left(\tproof p\arrow\tproof q\right)\arrow\tproof{\left(\appterm{\appterm{\const{imp}}p}q\right)}}\\
 & = & \lamterm{p,q}{\tterm{\tbool}}{\lamterm h{\left(\tproof p\arrow\tproof q\right)}h}\\
\const{imp\_elim} & : & \prodterm{p,q}{\tterm{\tbool}}{\tproof{\left(\appterm{\appterm{\const{imp}}p}q\right)}\arrow\tproof p\arrow\tproof q}\\
 & = & \lamterm{p,q}{\tterm{\tbool}}{\lamterm h{\tproof{\left(\appterm{\appterm{\const{imp}}p}q\right)}}{\lamterm x{\tproof p}{\appterm hx}}}
\end{array}
\]
By translating the introduction rules as $\lambda$-abstractions,
and the elimination rules as applications, we recover the reduction
of the proof terms, which corresponds to \emph{cut elimination} in
the original proofs.

As for equality, it is also possible to define it in terms of these
connectives. For example, we could use the Leibniz definition of equality,
which is the one used by Coq:
\[
\begin{array}{lcl}
\const{eq} & : & \prodterm{\alpha}{\ttype}{\tterm{\left(\tarrow{\alpha}{\left(\tarrow{\alpha}{\tbool}\right)}\right)}}\\
 & = & \lamterm{\alpha}{\ttype}{\lamterm x{\tterm{\alpha}}{\lamterm y{\tterm{\alpha}}{}}}\\
 &  & \appterm{\appterm{\const{forall}}{\left(\tarrow{\alpha}{\tbool}\right)}}{\left(\prodterm p{\tterm{\left(\tarrow{\alpha}{\tbool}\right)}}{\appterm{\appterm{\const{imp}}{\left(\appterm px\right)}}{\left(\appterm py\right)}}\right)}
\end{array}
\]
We would still need to assume some axioms to prove all the rules of
OpenTheory, namely $\const{FunExt}$ and $\const{PropExt}$ \cite{keller_importing_2010},
but at least this definition is closer to that of Coq. Since the $\lambda_{HOL}$
PTS is a strict subset of the calculus of inductive constructions,
we can adapt our translation to inject HOL directly into an embedding
of Coq in Dedukti \cite{boespflug_coqine:_2012}, or to combine HOL
proofs with Coq proofs in Dedukti \cite{assaf_mixing_2015}. Further
research into ways to eliminate these axioms (and thus maintain the
constructive aspect) when possible is the subject of ongoing work.

\section{\label{sec:conclusion}Conclusion}

We showed how to translate HOL to Dedukti by adapting techniques from
Cousineau and Dowek \cite{cousineau_embedding_2007} to define an
embedding that is sound, complete, and reduction preserving. Using
our implementation, we were able to translate the OpenTheory standard
library and verify it in Dedukti.

\subsubsection*{Future work}

The translation we have presented can be improved in several ways.
The current implementation suffers from a lack of linking: if a package
makes use of a type, constant, or theorem defined in another package,
we do not have a reference to the original definition. This is due
to a limitation of the OpenTheory article format. In OpenTheory, this
problem is resolved by adding a theory management layer, which is
responsible for composing and linking theories together \cite{hurd_opentheory_2011}.
It would be beneficial to integrate this layer in our translation
so that we can properly link the resulting files together.

While we used several optimizations including term sharing in our
implementation, there is still room for reducing the time and memory
consumption of the translation and the size of the generated files.
The caching techniques of Kaliszyk and Krauss \cite{kaliszyk_scalable_2013}
could be used in this regard to handle larger libraries and formalizations.

Finally, we can study how to combine the proofs obtained by this translation
with the proofs obtained from the translation of Coq \cite{boespflug_coqine:_2012}.
That will require a careful examination of the compatibility of the
two embeddings. First, the types of the two theories must coincide,
so that a natural number from HOL is the same as a natural number
from Coq for example. Second, we must make sure that the resulting
theory is consistent. For instance, we know that every type in HOL
is inhabited, which is inconsistent with the existence of empty types
in Coq, so we will need to modify the translations to avoid this.
A solution is to parameterize each HOL type variable by a witness
ensuring that it is non-empty. Our translation can be adapted for
this solution without much trouble. Some work has already been done
in this direction \cite{assaf_mixing_2015}.

\subsubsection*{Acknowledgments}

We thank Gilles Dowek for his support, as well as Mathieu Boespflug
and Chantal Keller for their comments and suggestions.

\bibliographystyle{eptcs}
\bibliography{biblio}

\end{document}